\documentclass{article}
\usepackage{amstext,amsgen,latexsym,amsmath}\usepackage{amssymb,amsfonts}
\usepackage{theorem}\usepackage{pifont}\usepackage[dvips]{graphics,epsfig}

 \newcommand{\bs}{\bigskip} 
  
 \newcommand{\hs}[1]{\hspace*{ #1 mm}}

 \def\bbox{\vrule height6pt width6pt depth1pt}

\theoremstyle{plain}
\theoremheaderfont{\bfseries}
 \newtheorem{theorem}{Theorem}[section] \newtheorem{lemma}[theorem]{Lemma}
 
 \newtheorem{corollary}[theorem]{Corollary}

 \newenvironment{proof}{\par \noindent
            {\bf Proof. \hs{2}}}{\hfill$\Box$ \vspace*{3mm}}
 
 \newenvironment{proofof}[1]{\vspace*{5mm} \par \noindent
         {\bf Proof of #1.\hs{2}}}{\hfill$\Box$ \vspace*{3mm}}

\newcommand{\ignore}[1]{}

\begin{document} 
\pagestyle{plain}

\begin{center}
{\Large {\bf Unbounded-Error One-Way Classical and Quantum 

\medskip 

Communication Complexity}}
\bs

\begin{center}
{\sc Kazuo Iwama}$^1$\footnote{Supported in part by Scientific Research Grant, Ministry of Japan, 16092101.} 
\hspace{5mm} {\sc Harumichi Nishimura}$^2$\footnote{Supported in part by Scientific Research Grant, 
Ministry of Japan, 18244210.}
\hspace{5mm} {\sc Rudy Raymond}$^3$ \hspace{5mm} {\sc Shigeru Yamashita}$^4$\footnote{Supported in part by Scientific Research Grant, Ministry of Japan, 16092218 and 19700010.} 
\end{center}

$^1${School of Informatics, Kyoto University}

{\tt iwama@kuis.kyoto-u.ac.jp} 

$^2${School of Science, Osaka Prefecture University}

{\tt hnishimura@mi.s.osakafu-u.ac.jp}

$^3${Tokyo Research Laboratory, IBM Japan}

{\tt raymond@jp.ibm.com}

$^4${Graduate School of Information Science, Nara Institute of Science and Technology} 

{\tt ger@is.naist.jp}. 
\end{center}
\bs

\begin{abstract}
This paper studies the gap between quantum one-way communication complexity $Q(f)$ 
and its classical counterpart $C(f)$, under the {\em unbounded-error} setting, i.e., 
it is enough that the success probability is strictly greater than $1/2$. 
It is proved that for {\em any} (total or partial) Boolean function $f$, 
$Q(f)=\lceil C(f)/2 \rceil$, i.e., the former is always exactly 
one half as large as the latter. The result has an application to obtaining  
(again an exact) bound for the existence of $(m,n,p)$-QRAC which is the $n$-qubit 
random access coding that can recover any one of $m$ original bits with success probability $\geq p$. 
We can prove that $(m,n,>1/2)$-QRAC exists if and only if $m\leq 
 2^{2n}-1$. Previously, only the construction of QRAC using one qubit, 
 the existence of $(O(n),n,>1/2)$-RAC, and the non-existence of $(2^{2n},n,>1/2)$-QRAC were known. 
\end{abstract}
\section{Introduction}\label{sec:introduction}

{\em Communication complexity} is probably the most popular model for studying the performance gap 
between classical and quantum computations. Even if restricted to the one-way private-coin setting 
(which means no shared randomness or entanglement), several interesting developments have been reported 
in the last couple of years. For {\em promise problems}, i.e., 
if we are allowed to use the fact that inputs to Alice and Bob 
satisfy some special property, exponential gaps are known: 
Bar-Yossef, Jayram and Kerenidis \cite{BJK04} constructed a relation 
to provide an exponential gap, $\Theta(\log n)$ vs.\ $\Theta(\sqrt{n})$, 
between one-way quantum and classical communication complexities. 
Recently, Gavinsky et al.\ \cite{GKKRW06} showed that 
a similar exponential gap also exists for a {\em partial} Boolean function.  

For {\em total} Boolean functions, i.e., if there is no available promise, 
there are no known exponential or even non-linear gaps: 
As mentioned in \cite{Aar04}, the equality function is a total Boolean function 
for which the one-way quantum communication complexity is approximately one half, 
$(1/2+o(1))\log n$ vs.\ $(1-o(1))\log n$, of the classical counterpart. 
This is the largest known gap so far. On the other hand, there are total 
Boolean functions for which virtually no gap exists between quantum and 
classical communication complexities. For example, those complexity gaps 
are only a smaller order additive term, $(1-H(p))n$ vs.\ 
$(1-H(p))n+O(\log n)$, for the index function \cite{ANTV99,Nay99}, and 
$n-2\log\frac{1}{2p-1}$ \cite{NS06} vs.\ $n - O(\log\frac{1}{2p-1})$ 
\cite{KN97} for the inner product function, where $p$ is the success probability. Note 
that all the results so far mentioned are obtained under the {\em bounded-error} 
assumption, i.e., the success probability must be at least $1/2+\alpha$ 
for some constant $\alpha$, being independent of the size of Boolean 
functions.  

Thus there seem to be a lot of varieties, depending on specific Boolean functions, 
in the quantum/classical gap of one-way communication complexity. 
In this paper it is shown that such varieties completely disappear if we use 
the {\em unbounded-error} model where it is enough that the success probability 
is strictly greater than $1/2$.   

\subsection{Our Contribution}

We show that one-way quantum communication complexity of any (total or partial) Boolean function 
is always exactly (without an error of even $\pm 1$) one half of the one-way classical communication complexity 
in the {\em unbounded-error} setting. 
The study of unbounded-error (classical) communication complexity was initiated by Paturi and Simon \cite{PS86}. 
They characterized almost tightly the unbounded-error one-way communication complexity of Boolean function $f$, 
denoted by $C(f)$, in terms of a geometrical measure $k_f$ which is the minimum dimension of the arrangement of 
points and hyperplanes. Namely, they proved that $\lceil\log k_f\rceil\leq C(f)\leq \lceil \log k_f\rceil+1$. 
We show that such a characterization is also applicable to the unbounded-error one-way 
{\em quantum} communication complexity $Q(f)$. 
To this end, we need to link accurately the one-way quantum 
communication protocol to the arrangement of points and hyperplanes, 
which turns out to be possible using geometric facts on quantum states \cite{JS01,KK04}. 
As a result we show that $Q(f)=\lceil\log (k_f+1)/2\rceil$. 
Moreover, we also remove the small gap in \cite{PS86}, proving $C(f)=\lceil\log (k_f+1)\rceil$. 
This enables us to provide the exact relation between $Q(f)$ and $C(f)$, i.e., 
$Q(f)=\lceil C(f)/2\rceil$. 

Our characterizations of $Q(f)$ and $C(f)$ have an application to quantum random access coding 
(QRAC) and classical random access coding (RAC) introduced by Ambainis et al.\ \cite{ANTV99}. 
The $(m,n,p)$-QRAC (resp.\ $(m,n,p)$-RAC) is the $n$-qubit (resp.\ $n$-bit) coding that 
can recover any one of $m$ bits with success probability $\geq p$. 
The asymptotic relation among the three parameters $m,n,p$ was shown 
in \cite{ANTV99} and \cite{Nay99}: If $(m,n,p)$-QRAC exists, then $n\geq (1-H(p))m$, 
while there exists $(m,n,p)$-RAC if $n \leq (1-H(p))m+O(\log m)$. 
This relation gives us a tight bound on $n$ when $p$ is relatively far from $1/2$. 
Unfortunately these inequalities give us little information under the unbounded-error setting 
or when $p$ is very close to $1/2$, because the value of $(1-H(p))m$ become less than one. 
Hayashi et al. \cite{HINRY06} showed that $(m,n,p)$-QRAC with $p>1/2$ does not 
exist when $m=2^{2n}$. Our characterization directly shows that this is tight, that is, 
$(m,n,>1/2)$-QRAC exists if and only if $m\leq 2^{2n}-1$, which solves the remained open problem 
in \cite{HINRY06}. A similar tight result on the existence of $(m,n,>1/2)$-RAC is also obtained 
from our characterization. Moreover, we also give concrete constructions of such QRAC and RAC 
with an analysis of their success probability. 
 
\subsection{Related Work}
We mainly focus on the gap between classical and quantum communication complexities. 

{\bf Partial/Total Boolean Functions.} For total functions, the one-way quantum communication complexity 
is nicely characterized or bounded below in several ways. 
Klauck \cite{Kla00} characterized the one-way communication complexity of total Boolean functions 
by the number of different rows of the communication matrix in the {\em exact} setting, 
i.e., the success probability is one, and showed that it equals to 
the one-way deterministic communication complexity. Also, he gave a lower bound of bounded-error 
one-way quantum communication complexity of total Boolean functions by the VC dimension. 
Aaronson \cite{Aar04,Aar06} presented lower bounds of the one-way quantum communication complexity 
that are also applicable for partial Boolean functions. His lower bounds are given in terms 
of the deterministic or bounded-error classical communication complexity 
and the length of Bob's input, which are shown to be tight by using 
the partial Boolean function of Gavinsky et al.\ \cite{GKKRW06}. 

{\bf One-way/Two-way/SMP Models.} 
Two-way communication model is also popular. It is known that the two-way communication complexity 
has a non-linear quantum/classical gap for total functions in the bounded-error model. 
The current largest gap is quadratic. Buhrman, Cleve and Wigderson \cite{BCW98} 
showed that the almost quadratic gap, $O(\sqrt{n}\log n)$ vs.\ $\Omega(n)$, 
exists for the disjointness function. This gap was improved to $O(\sqrt{n})$ vs.\ $\Omega(n)$ in \cite{AA05}, 
which turned out to be optimal within a constant factor for the disjointness function \cite{Razb03}. 
On the contrary, in the unbounded-error setting, two-way communication model 
can be simulated by one-way model with only one bit additional communication \cite{PS86}. 
In the simultaneous message passing (SMP) model where we have a referee other than Alice and Bob, 
an exponential quantum/classical gap for total functions was shown by Buhrman et al.\ \cite{BCWW01}. 

{\bf Private-coin/Public-coin Models.} 
The exponential quantum/classical separations in \cite{BJK04} and \cite{GKKRW06} 
still hold under the public-coin model where Alice and Bob share random coins, 
since the one-way classical public-coin model can be simulated by 
the one-way classical private-coin model 
with additional $O(\log n)$-bit communication \cite{New91}. 
However, exponential quantum/classical separation for total functions remains open 
for all of the bounded-error two-way, one-way and SMP models. 
Note that the public-coin model is too powerful in the unbounded-error model: we can easily see that 
the unbounded-error one-way (classical or quantum) communication complexity of any function (or relation) 
is $1$ with prior shared randomness.  

{\bf Unbounded-error Models.} 
Since the seminal paper \cite{PS86}, the unbounded-error (classical) 
one-way communication complexity has been developed in the literature 
\cite{Fors02,FKLMSS01,FS06}. (Note that in the classical setting, the 
difference of communication cost between one-way and two-way models is 
at most $1$ bit.) Klauck \cite{Kla01} also studied a variant of the 
unbounded-error quantum and classical communication complexity, called 
the weakly unbounded-error communication complexity: the cost is 
communication (qu)bits plus $\log 1/\epsilon$ where $1/2+\epsilon$ is 
the success probability. He characterized the discrepancy, a useful 
measure for bounded-error communication complexity \cite{KN97}, in terms 
of the weakly unbounded-error communication complexity.

\section{Preliminaries}\label{sec:preliminaries}

For basic notations of quantum computing, see \cite{NC00}. In this 
paper, a ``function'' represents both total and partial Boolean functions. 

{\bf Communication Complexity.} 
The two-party communication complexity model is defined as follows.
One party, say Alice, has input $x$ from a finite set $X$ and another party, say Bob, input $y$ 
from a finite set $Y$. One of them, say, Bob wants to compute the value $f(x,y)$ for a function $f$. (In some cases, 
relations are considered instead of functions.) Their communication 
process is called a {\em quantum (resp.\ classical) protocol} if the 
communication is done by using quantum bits (resp.\ classical bits). In 
particular, the protocol is called {\em one-way} if the communication is 
only from Alice to Bob. The communication cost of the protocol is 
the maximum number of (qu)bits needed over all $(x,y)\in X\times Y$ by 
the protocol. The {\em unbounded-error} one-way quantum (resp.\ classical) communication complexity of $f$, 
denoted by $Q(f)$ (resp.\ $C(f)$), is the communication cost of the best one-way quantum 
(resp.\ classical) protocol with success probability strictly larger than $1/2$.  
In what follows, the term ``classical'' is often omitted when it is 
clear from the context. We denote the communication matrix of $f$ by $\pmb{M}_f=((-1)^{f(x,y)})$. 
(We use the bold font letters for denoting vectors and matrices.)

{\bf Arrangements.}
The notion of arrangement has often been used as one of the basic concepts in computer science 
such as computational geometry and learning theory. The arrangement of points and hyperplanes 
has two well-studied measures: the minimum dimension and margin complexity. 
We use the former, as in \cite{PS86}, to characterize the unbounded-error one-way communication complexity 
(while the latter was used in \cite{GKW06} to give a lower bound of bounded-error 
quantum communication complexity under prior shared entanglement). 
A point in $\mathbb{R}^n$ is denoted by the corresponding $n$-dimensional real vector. 
Also, a hyperplane $\{(a_i)\in\mathbb{R}^n\mid \sum_{i=1}^n a_ih_i=h_{n+1}\}$ 
on $\mathbb{R}^n$ is denoted by the $(n+1)$-dimensional real vector 
${\mathbf h}=(h_1,\ldots,h_n,h_{n+1})$, meaning that any point $(a_i)$ on the plane 
satisfies the equation $\sum_{i=1}^n a_ih_i =h_{n+1}$. 
A $\{1,-1\}$-valued matrix $\pmb{M}$ on $X \times Y$ is {\it realizable by an arrangement} 
of a set of $|X|$ points ${\mathbf p}_x=(p_{1}^x,\ldots,p_{k}^x)$ 
and a set of $|Y|$ hyperplanes ${\mathbf h}_y =(h_{1}^y,\ldots,h_{k}^y,h_{k+1}^y)$ in $\mathbb{R}^k$ 
if for any $x \in X$ and $y \in Y$, $\delta({\mathbf p}_x,{\mathbf h}_y):= 
\mbox{sign}(\sum_{i=1}^k {p_{i}^x h_{i}^y} - h_{k+1}^y)$ is equal to $\pmb{M}(x,y)$. 
Here, $\mbox{sign}(a)=1$ if $a>0$, $-1$ if $a<0$, and $0$ otherwise. Intuitively, 
the point lies above, below, or on the plane 
if $\delta({\mathbf p}_x,{\mathbf h}_y) =1$, $-1$, and $0$, respectively. 
The value $k$ is called the dimension of the arrangement. Let 
$k_{\pmb{M}}$ denote the smallest dimension of all arrangements that 
realize $\pmb{M}$. In particular, if $\pmb{M}=\pmb{M}_f$ then we denote 
$k_{\pmb{M}}$ by $k_f$, and say that $f$ is {\it realized} by the arrangement.

{\bf Bloch Vector Representations of Quantum States.} 
Mathematically, the $N$-level quantum state is represented 
by an $N\times N$ positive matrix $\pmb{\rho}$ satisfying $\mbox{Tr}(\pmb{\rho})=1$. 
(Note that if $N=2^n$ then $\pmb{\rho}$ is considered as a quantum state that consists of $n$ qubits.) 
In this paper we use $N\times N$ matrices $\pmb{I}_N,\pmb{\lambda}_1,\ldots,\pmb{\lambda}_{N^2-1}$, 
called {\em generator matrices}, as a basis to represent $N$-level 
quantum states. Here, $\pmb{I}_N$ is the identity matrix (the subscript $N$ is often omitted), 
and $\pmb{\lambda}_i$'s are the generators of $SU(N)$ 
satisfying (i) $\pmb{\lambda}_i = \pmb{\lambda}_i^{\dagger}$, 
(ii) $\mbox{Tr}(\pmb{\lambda}_i) = 0$ 
and (iii) $\mbox{Tr}(\pmb{\lambda}_i \pmb{\lambda}_j) = 2\delta_{ij}$. 
Then, the following lemma is known (see, e.g., \cite{KK04}). 

\begin{lemma}\label{mstate_vec}
For any $N$-level quantum state $\pmb{\rho}$ and any $N\times N$ generator matrices $\pmb{\lambda}_i$'s, 
there exists an $(N^2-1)$-dimensional vector ${\mathbf r}=(r_i)$ such that $\pmb{\rho}$  
can be written as 
\begin{equation}\label{mstate_vec_eq}
\pmb{\rho}= \frac{1}{N}\left(\pmb{I} + 
\sqrt{\frac{N(N-1)}{2}}\sum_{i=1}^{N^2-1}r_i \pmb{\lambda}_i\right).
\end{equation}
\end{lemma}

The vector ${\mathbf r}$ in this lemma is often called the {\em Bloch vector} of $\pmb{\rho}$. 
Note that $\pmb{\lambda}_i$ can be any generator matrices satisfying the 
above conditions. In particular, it is well-known \cite{NC00} that for $N = 2$ one can choose $\pmb{\sigma}_1=
\left(\begin{array}{cc}
1 & 0\\
0 & -1
\end{array}
\right)$, 
$\pmb{\sigma}_2=
\left(\begin{array}{cc}
0 & 1\\
1 & 0
\end{array}
\right)$, and 
$\pmb{\sigma}_3=
\left(\begin{array}{cc}
0 & -\imath\\
\imath & 0
\end{array}
\right)$ 
of Pauli matrices as $\pmb{\lambda}_1,\pmb{\lambda}_2$, and $\pmb{\lambda}_3$, respectively. 
Generally for $N=2^n$, one can choose the tensor products of Pauli 
matrices, including $\pmb{I}$, for $\pmb{\lambda}_1,\ldots,\pmb{\lambda}_{N^2-1}$. 

Note that Lemma~\ref{mstate_vec} is a necessary condition for $\pmb{\rho}$ 
to be a quantum state. Although our knowledge of the sufficient condition 
is relatively weak (say, see \cite{JS01,KK04}), the following two lemmas 
on the mathematical description of $N$-level quantum states are enough for our purpose. 

\begin{lemma}[\cite{KK04}]\label{kk04lemma}
Let $r=\sqrt{\sum_{i=1}^{N^2-1}r_i^2}$. Then, $\pmb{\rho}
=\frac{1}{N} \left(\pmb{I} + \sqrt{\frac{N(N-1)}{2}} \sum_{i=1}^{N^2-1} r_i\pmb{\lambda}_i\right)$ 
is a quantum state if and only if 
$r \leq \sqrt{\frac{2}{N(N-1)}}\frac{1}{|m(\sum_{i=1}^{N^2-1}\left(\frac{r_i}{r}\right)\pmb{\lambda}_i)|}$, 
where $m(\pmb{A})$ denotes the minimum of eigenvalues of a matrix $\pmb{A}$, and $\pmb{\lambda}_i$'s 
are any generator matrices. 
\end{lemma}

\begin{lemma}[\cite{JS01}]\label{balls}
Let ${B}(\mathbb{R}^{N^2-1})$ be the set of Bloch vectors 
of all $N$-level quantum states. Let $D_{r_s}(\mathbb{R}^{N^2-1})=\{{\mathbf r}\in\mathbb{R}^{N^2-1} 
\mid |{\mathbf r}| \le \frac{1}{{N-1}}\}$ (called the small ball), 
and $D_{r_l}({\mathbf R}^{N^2-1})=\{{\mathbf r}\in\mathbb{R}^{N^2-1} 
\mid |{\mathbf r}| \le 1 \}$ (called the large ball). 
Then, $D_{r_s}(\mathbb{R}^{N^2-1})\subseteq{B}(\mathbb{R}^{N^2-1})
\subseteq D_{r_l}(\mathbb{R}^{N^2-1})$. 
\end{lemma}

\section{Quantum Tight Bound}

In \cite{HINRY06}, we gave a geometric view of the quantum protocol on random access coding. 
It turns out that this view together with the notion of arrangements is a powerful tool 
for characterizing the unbounded-error one-way quantum communication complexity.

\begin{theorem}\label{qcc}
$Q(f)= \lceil\log(k_f+1)/2\rceil$ for every function $f:X\times Y\rightarrow \{0,1\}$.
\end{theorem}

The outline of the proof is as follows: In Lemma~\ref{neccesary} we 
first establish a relation similar to Lemma~\ref{mstate_vec} between a 
POVM (Positive Operator-Valued Measure) $\{\pmb{E},\pmb{I-E}\}$  
over $n$ qubits and a $(2^{2n}-1)$-dimensional (Bloch) vector ${\mathbf h}(\pmb{E})$. 
Then, we prepare Lemma~\ref{trace_bloch} to show that the measurement results of POVM $\{\pmb{E},\pmb{I-E}\}$ 
on a state $\pmb{\rho}$ correspond to the arrangement operation 
$\delta({\mathbf r}(\pmb{\rho}),{\mathbf h}(\pmb{E}))$, where ${\mathbf r}(\pmb{\rho})$ 
is the Bloch vector for $\pmb{\rho}$.  

Now in order to prove $Q(f) \ge \lceil\log(k_f+1)/2\rceil$, 
suppose that there is a protocol whose communication complexity is $n$. 
This means for any $x \in X$ and $y \in Y$, we have $n$-qubit states 
$\pmb{\rho}_x$ and POVMs $\{\pmb{E}_y,\pmb{I}-\pmb{E}_y\}$ such that: 
(i) the dimensions of ${\mathbf r}(\pmb{\rho}_x)$ and ${\mathbf h}(\pmb{E}_y)$ 
are $2^{2n}-1$ and $2^{2n}$ (by Lemmas~\ref{mstate_vec} and 
\ref{neccesary}, and note that $N=2^n$), and (ii) $\pmb{M}_f(x,y)=\mbox{sign}(\mathrm{Tr}(\pmb{E}_y\pmb{\rho}_x)-1/2)
=\delta({\mathbf r}(\pmb{\rho}_x),{\mathbf h}(\pmb{E}_y))$ 
(the first equality by the assumption and the second one by Lemma \ref{trace_bloch}). 
By (ii) we can conclude that the arrangement of points ${\mathbf r}(\pmb{\rho}_x)$ 
and hyperplanes ${\mathbf h}(\pmb{E}_y)$ realizes $f$, and by (i) its dimension is $2^{2n}-1$. 
Thus, $k_f$ is at most $2^{2n}-1$, implying that $n$ ($=Q(f)$) $\ge \lceil\log(k_f+1)/2\rceil$. 

To prove the converse, suppose that there exists an $(N^2-1)$-dimensional 
arrangement of points ${\mathbf r}_x$ and hyperplanes ${\mathbf h}_y$ 
realizing $f$. For simplicity, suppose that $N^2-1=k_f$ (see the proof of Theorem \ref{qcc} for the details). 
Let us fix some generator matrices $\pmb{\lambda}_i$'s. However, 
$\pmb{\rho}_x$ obtained directly from $\pmb{\lambda}_i$'s and ${\mathbf r}_x$ 
by Eq.(\ref{mstate_vec_eq}) may not be a valid quantum state. 
Fortunately, by Lemma~\ref{embed_states} we can simply multiply ${\mathbf r}_x$ by a fixed constant 
factor to obtain ${\mathbf r}'_x$ such that ${\mathbf r}'_x$ lies in the 
small ball in Lemma~\ref{balls} and therefore corresponds to an $n$-qubit state $\pmb{\rho}({\mathbf r}_x')$. 
Similarly, by Lemma~\ref{embed_meas} 
we can get ${\mathbf h}'_y$ corresponding to POVM $\{\pmb{E}({\mathbf h}_y'),\pmb{I}-\pmb{E}({\mathbf h}_y')\}$. 
Obviously, the arrangement of points ${\mathbf r}_x'$ and hyperplanes ${\mathbf h}_y'$ realizes $f$, 
its dimension is the same $N^2-1$ and the corresponding $\pmb{\rho}({\mathbf r}'_x)$ and 
$\{\pmb{E}({\mathbf h}'_y),\pmb{I}-\pmb{E}({\mathbf h}'_y)\}$ are an $N$-level (or $\lceil\log N\rceil$-qubit) 
quantum state and a POVM over $N$-level quantum states, respectively. 
Now, by Lemma~\ref{trace_bloch}, we can compute $f(x,y)$ 
by $\mbox{sign}(\mathrm{Tr}(\pmb{E}({\mathbf h}_y')\pmb{\rho}({\mathbf r}_x'))-1/2)$, 
which means $Q(f) \leq \lceil\log N\rceil = \lceil\log(k_f+1)/2\rceil$. 

According to the above outline, we start to present technical lemmas 
whose details are omitted. The following lemma, shown similarly as Lemma \ref{mstate_vec}, 
is a necessary condition for $\{\pmb{E},\pmb{I}-\pmb{E}\}$ to be a POVM. 

\begin{lemma}\label{neccesary}
For any POVM $\{\pmb{E},\pmb{I}-\pmb{E}\}$ over $N$-level quantum states 
and $N\times N$ generator matrices $\pmb{\lambda}_i$'s, there exists 
an $N^2$-dimensional vector ${\mathbf e}=(e_i)$ such that 
$\pmb{E}$ can be written as $\pmb{E} = e_{N^2} \pmb{I} + \sum_{i=1}^{N^2-1} e_i\pmb{\lambda}_i$.  
\end{lemma} 

We call the above vector ${\mathbf e}=(e_1,\ldots,e_{N^2-1},e_{N^2})$ 
the {\em Bloch vector} of POVM $\{\pmb{E},\pmb{I}-\pmb{E}\}$. 
The next lemma relates the probability distribution 
of binary values obtained by measuring a quantum state $\pmb{\rho}$ 
with a POVM $\{\pmb{E},\pmb{I}-\pmb{E}\}$ with their Bloch vectors. 

\begin{lemma}\label{trace_bloch}
Let ${\mathbf r}=(r_i) \in {\mathbf R}^{N^2-1}$ and ${\mathbf e}=(e_i) \in {\mathbf R}^{N^2}$ 
be the Bloch vectors of an $N$-level quantum state $\pmb{\rho}$ 
and a POVM $\{\pmb{E},\pmb{I}-\pmb{E}\}$.  
Then, the probability that the measurement value $0$ is obtained is  
\[
\mathrm{Tr}\left(\pmb{E}\pmb{\rho}\right) = e_{N^2} + \sqrt{\frac{2(N-1)}{N}} \sum_{i=1}^{N^2-1} r_ie_i.
\]
\end{lemma}

The last two lemmas provide a shrink-and-shift mapping from any real 
vectors and hyperplanes to, respectively, Bloch vectors of quantum states lying in the small ball of 
Lemma~\ref{balls} and POVMs.

\begin{lemma}\label{embed_states}
(1) 
For any ${\mathbf r}=(r_1,r_2,\ldots,r_k)\in\mathbb{R}^k$ and $N$ satisfying $N^2\geq k+1$, 
\[
\pmb{\rho}({\mathbf r}) = \frac{1}{N}\left( \pmb{I} +\sqrt{\frac{N(N-1)}{2}} \sum_{i=1}^{k} 
\left(\frac{r_i}{|{\mathbf r}|(N-1)}\right) \pmb{\lambda}_i 
\right)
\]
is an $N$-level quantum state. 

(2) If $\pmb{\rho}({\mathbf r})$ is a quantum state, then $\pmb{\rho}(\gamma{\mathbf r})$ 
is also a quantum state for any $\gamma\leq 1$. 
\end{lemma}

\begin{lemma}\label{embed_meas}
For any hyperplane ${\mathbf h} = \left(h_1,\ldots,h_k,h_{k+1}\right)\in\mathbb{R}^{k+1}$, 
let $N$ be any number such that $N^2\geq k+1$, and let $\alpha,\beta$ be two positive 
numbers that are at most $\frac{1}{2\left(|h_{k+1}|+h\sqrt{\frac{2(N-1)}{N}}\right)}$ 
where $h=\sum_{i=1}^k h_i^2$. 
Then, the $N^2$-dimensional vector defined by 
$
{\mathbf h}(\alpha,\beta)=(\beta h_1,\ldots,\beta h_k,0,\ldots,0,1/2-\alpha h_{k+1})
$ is the Bloch vector of a POVM $\{\pmb{E}_0,\pmb{E}_1\}$ over $N$-level quantum states, 
where $\pmb{E}_0$ and $\pmb{E}_1$ are given as 
\begin{equation}\label{eq071-1}
\pmb{E}_0 = \left(\frac{1}{2} - \alpha h_{k+1}\right)\pmb{I} 
+ \beta \sum_{i = 1}^{k} h_i \pmb{\lambda}_i\ \ \mbox{and}
\ \ \pmb{E}_1 = \left(\frac{1}{2} + \alpha h_{k+1}\right)\pmb{I} 
- \beta \sum_{i = 1}^{k} h_i \pmb{\lambda}_i.
\end{equation}
\end{lemma}

Now we prove our main theorem in this section.

\begin{proofof}{Theorem \ref{qcc}}
$k_f$ is simply written as $k$ in this proof. 

$(Q(f) \ge \lceil \log(k+1)/2 \rceil)$. 
Let $n=Q(f)$ and $N=2^{n}$. Assume that there is an $n$-qubit protocol for $f$. 
That is, Alice on input $x$ sends an $n$-qubit state $\pmb{\rho}_x$ to Bob with input $y$.  
He then measures $\pmb{\rho}_x$ with a POVM $\{\pmb{E}_y,\pmb{I}-\pmb{E}_y\}$ 
so that $\mbox{sign}(\mbox{Tr}(\pmb{E}_y\pmb{\rho}_x) - 1/2) = \pmb{M}_f(x,y)$. 
From Lemmas \ref{mstate_vec} and \ref{neccesary}  
we can define the points ${\mathbf p}_x=(p_{i}^x) \in \mathbb{R}^{N^2-1}$ 
and hyperplanes ${\mathbf h}_y=(h_{i}^y)\in \mathbb{R}^{N^2}$ 
so that ${\mathbf p}_x$ is the Bloch vector of $\pmb{\rho}_x$, 
and ${\mathbf h}_y =\left(\sqrt{\frac{2(N-1)}{N}}e_{1}^y,
\ldots,\sqrt{\frac{2(N-1)}{N}}
e_{N^2-1}^y,1/2-e_{N^2}^y\right)$ 
where ${\mathbf e}_y=(e_{i}^y)$ is the Bloch vector of the POVM $\{\pmb{E}_y,\pmb{I}-\pmb{E}_y\}$. 
Notice that by Lemma~\ref{trace_bloch}, $\mbox{Tr}(\pmb{E}_y\pmb{\rho}_x)
=e_{N^2}^y+\sqrt{\frac{2(N-1)}{N}}\sum_{i=1}^{N^2-1} p_{i}^x e_{i}^y$, 
which is $> 1/2$ if $\pmb{M}_f(x,y)= 1$ and $< 1/2$ if $\pmb{M}_f(x,y)= -1$ by assumption. 
Thus, we can see that 
\[
\delta({\mathbf p}_x,{\mathbf h}_y)
=\mbox{sign}\left(e_{N^2}^y+\sqrt{\frac{2(N-1)}{N}}\sum_{i=1}^{N^2-1} p_{i}^x e_{i}^y-1/2\right)
=\pmb{M}_f(x,y),
\]
meaning that there exists an arrangement of points and hyperplanes 
in $\mathbb{R}^{N^2-1}$ which realizes $f$. 
Thus, by definition, $k$ is at most $N^2-1=2^{2n}-1$ which implies 
$Q(f)=n \geq\lceil \log(k+1)/2 \rceil$. 

($Q(f) \le \lceil \log(k+1)/2 \rceil$). 
Suppose that there is a $k$-dimensional arrangement of points ${\mathbf p}_x=(p_{i}^x)\in\mathbb{R}^k$ 
and hyperplanes ${\mathbf h}_y=(h_{i}^y)\in\mathbb{R}^{k+1}$ that realizes $\pmb{M}_f$. 
That is, $\delta({\mathbf p}_x,{\mathbf h}_y) = \pmb{M}_f(x,y)$ for every $(x,y)\in X\times Y$. 
 By carefully shrinking-and-shifting this arrangement into Bloch vectors 
 in the small ball, we will show the construction an $n$-qubit protocol for $f$, that is, $n$-qubit states 
 $\pmb{\rho}_x$ for Alice and POVMs $\{\pmb{E}_y,\pmb{I}-\pmb{E}_y\}$ 
 for Bob with the smallest $n$ satisfying $k\leq 2^{2n}-1$, and hence obtain $Q(f)\leq n= \lceil \log(k+1)/2 \rceil$.  

Let $\gamma_x=\mbox{min}\left\{\frac{1}{|{\mathbf p}_x|(2^n-1)},\frac{1}{2^n-1}\right\}$ for each $x\in X$. 
Then, since $(2^n)^2\geq k+1$, Lemma~\ref{embed_states} implies that 
$\frac{1}{2^n}\left(\pmb{I}+\sqrt{ \frac{2^n(2^n-1)}{2} }\sum_{i=1}^{k}\gamma_x p_{i}^x \pmb{\lambda}_i \right)$ 
is an $n$-qubit state, and hence $\gamma_x {\mathbf p}_x$ is the Bloch vector of its qubit state.  
Moreover, Lemma~\ref{embed_meas} implies that by taking 
$\beta_y=\frac{1}{2\left(|h_{k+1}^y|+\sqrt{\sum_{i=1}^{k}(h_{i}^y)^2}\sqrt{\frac{2(2^n-1)}{2^n} }\right)}$, 
${\mathbf h}_y(\beta_y,\beta_y)=(\beta_y h_{1}^y,\ldots,\beta_y h_{k}^y,0,$ $\ldots,0,1/2-\beta_y h_{k+1}^y)$ 
is the Bloch vector of a POVM over $n$-qubit states. 

Now let $\gamma=\frac{1}{\sqrt{2}}\mbox{min}_{x\in X} \gamma_x$, $\beta=\mbox{min}_{y\in Y} \beta_y$, 
and $\alpha=\sqrt{\frac{2(2^n-1)}{2^n}}\gamma\beta$. 
Since $\gamma\leq \gamma_x$ for any $x\in X$ and $0 < \alpha<\beta\leq \beta_y$ for any $y\in Y$, 
Lemmas \ref{embed_states}(2) and \ref{embed_meas} show that $\gamma{\mathbf p}_x$ and 
${\mathbf h}_y(\beta,\alpha)$ are also the Bloch vectors of an $n$-qubit state $\pmb{\rho}_x$ 
and a POVM $\{\pmb{E}_y,\pmb{I}-\pmb{E}_y\}$ over $n$-qubit states, respectively. 
By Lemma~\ref{trace_bloch}, the probability that the measurement value $0$ is obtained is 
\begin{eqnarray*}
\mbox{Tr}(\pmb{E}_y\pmb{\rho}_x) 
&=& \frac{1}{2} - \alpha h_{k+1}^y + \sqrt{\frac{2(2^n-1)}{2^n}}\gamma \beta \sum_{i=1}^k p_{i}^x h_{i}^y
= \frac{1}{2} + \alpha \left(\sum_{i=1}^k p_{i}^x h_{i}^y - h_{k+1}^y \right)
\\
&=& \left\{
\begin{array}{l}
> 1/2 \mbox{~if~} \pmb{M}_f(x,y) = 1\\
< 1/2 \mbox{~if~} \pmb{M}_f(x,y)= -1,
\end{array}
\right.
\end{eqnarray*}
where the last inequality comes from the assumption. Therefore, the 
 states $\pmb{\rho}_x$ and POVMs $\{\pmb{E}_y, \pmb{I}-\pmb{E}_y \}$ can 
 be used to obtain an $n$-qubit protocol for $f$.
\end{proofof}

Combined with the results in \cite{Fors02,FS06}, Theorem \ref{qcc} gives 
us a nontrivial bound for the inner product function $IP_n$ (i.e., 
$IP_n(x,y)=\sum_{i=1}^n x_iy_i$ mod $2$ for any $x=x_1\cdots 
x_n\in\{0,1\}^n$ and $y=y_1\cdots y_n\in\{0,1\}^n$). Note that the 
bounded-error quantum communication complexity is at least $n-O(1)$, and 
$n/2-O(1)$ even if we allow two-way protocol and prior entanglement 
\cite{NS06}. 

\begin{corollary}\label{ip-bound}
$\lceil n/4 \rceil \leq Q(IP_n)
\leq \lceil ((\log 3) n+2)/4 \rceil$. 
\end{corollary}

\section{Classical Tight Bound} 
Paturi and Simon \cite{PS86} shows that for every function $f:X\times Y\rightarrow\{0,1\}$, 
$\lceil\log k_f\rceil\leq C(f)\leq \lceil\log k_f\rceil+1$. We remove 
this small gap as follows.
 
\begin{theorem}\label{ccc}
$C(f)= \lceil \log(k_f+1) \rceil$ for every function $f:X\times Y\rightarrow\{0,1\}$.
\end{theorem}

\begin{proof}
Let $k=k_f$ in this proof. 

($C(f) \ge \lceil\log{(k+1)}\rceil$). Let $N=2^{C(f)}$.  
Suppose that there is a $C(f)$-bit protocol for $f$. Paturi and Simon (in Theorem~2 in \cite{PS86}) 
gave an $N$-dimensional arrangement of points ${\mathbf p}_x=(p_{i}^x)\in\mathbb{R}^N$ and hyperplanes 
${\mathbf h}_y=(h_{1}^y,\ldots,h_{N}^y,1/2)\in \mathbb{R}^{N+1}$, that is, 
$\delta({\mathbf p}_x,{\mathbf h}_y) = \pmb{M}_f(x,y)$ for every $(x,y)\in X\times Y$. 
Noting that the points ${\mathbf p}_x$ are probabilistic vectors 
 satisfying $\sum_{i=1}^{N} p_i=1$, we can reduce the dimension of the arrangement to $N-1$. 
We define ${\mathbf q}_x=(q_{i}^x)\in\mathbb{R}^{N-1}$ and ${\mathbf l}_y=(l_{i}^y)\in\mathbb{R}^{N}$ 
from ${\mathbf p}_x$ and ${\mathbf h}_y$, respectively, as follows: 
${\mathbf q}_x= \left(p_{1}^x,p_{2}^x,\ldots,p_{N-1}^x \right)$ and 
${\mathbf l}_y= \left(h_{1}^y-h_{N}^y, h_{2}^y-h_{N}^y, \ldots, h_{N-1}^y-h_{N}^y, \frac{1}{2}- h_{N}^y \right)$. 
From the assumption and $p_{N}^x = 1 - \sum_{i}^{N-1} p_{i}^x$, 
\begin{align*}
\sum_{i=1}^{N-1} q_{i}^x l_{i}^y - l_{N}^y
&= \sum_{i=1}^{N-1} p_{i}^x (h_{i}^y -h_{N}^y) -\frac{1}{2}+h_{N}^y
= \sum_{i=1}^{N-1} p_{i}^x h_{i}^y -\frac{1}{2}+h_{N}^y-\sum_{i=1}^{N-1}p_{i}^x h_{N}^y \\
&= \sum_{i=1}^{N} p_{i}^x h_{i}^y -\frac{1}{2}
=\left\{
\begin{array}{ll}
> 0\ \mbox{if}\ \pmb{M}_f(x,y)=1\\
< 0\ \mbox{if}\ \pmb{M}_f(x,y)=-1.
\end{array}
\right.
\end{align*}
Thus, $\delta({\mathbf q}_x,{\mathbf l}_y)=\pmb{M}_f(x,y)$ for every $(x,y)\in X\times Y$. 
That is, $M_f$ is realizable by the $(N-1)$-dimensional arrangement of points ${\mathbf q}_x$ 
and hyperplanes ${\mathbf l}_y$. By definition, $k\leq N-1=2^{C(f)} - 1$, which means that $C(f)
\geq \lceil\log{(k+1)}\rceil$.

($C(f) \le \left\lceil\log{(k+1)}\right\rceil$). The proof is also based on that of Theorem~2 of Paturi and Simon \cite{PS86}. 
They showed the existence of a protocol where Alice (with input $x$) 
sends a probabilistic mixture of (at most) $k+2$ different messages to Bob (with input $y$). 
In this proof we reduce the number of messages to $k+1$. That is, we construct the following protocol using $k+1$ 
different messages: Alice sends a message $S_j$ with probability $q_{j}^x$ where $j \in [k+1]$, 
and Bob outputs $0$ with probability $l_{j}^y$ upon receiving $S_j$. Here, $[n]:=\{1,2,\ldots,n\}$ 
for any $n\in\mathbb{N}$. We will show that the probability of Bob outputs $0$, 
represented as $\sum_{j=1}^{k+1}q_{j}^x l_{j}^y$, is $>1/2$ if $\pmb{M}_f(x,y)=1$ and $<1/2$ if $\pmb{M}_f(x,y)=-1$. 

Assume that there exists a $k$-dimensional arrangement of points ${\mathbf p}_x=(p_{i}^x)\in\mathbb{R}^k$ 
and hyperplanes ${\mathbf h}_y=(h_{i}^y)\in \mathbb{R}^{k+1}$ that realizes $\pmb{M}_f$, 
that is, $\delta({\mathbf p}_x,{\mathbf h}_y) = \pmb{M}_f(x,y)$ for every $(x,y)\in X\times Y$. 
Let $s =\mbox{max}_{x\in X}\mbox{max}_{i\in[k]} |p_{i}^x|$, $\alpha_x = 1 + \sum_{i=1}^k (s + p_{i}^x)$ for each $x \in X$, 
and $\beta_y = \mbox{max}(|h_{1}^y|,\ldots,|h_{k}^y|, |h_{k+1}^y +s\sum_{i=1}^k h_{i}^y|)$ for each $y \in Y$. 
Then, we define ${\mathbf q}_x=(q_{i}^x)\in\mathbb{R}^{k+1}$ and ${\mathbf l}_y=(l_{i}^y)\in\mathbb{R}^{k+1}$ by 
${\mathbf q}_x=\left(\frac{s+p_{1}^x}{\alpha_x},\frac{s+p_{2}^x}{\alpha_x},
\ldots,\frac{s+p_{k}^x}{\alpha_x},\frac{1}{\alpha_x}\right)$ and ${\mathbf l}_y 
= 
\left(\frac{1}{2} + \frac{h_{1}^y}{2\beta_y},
\frac{1}{2} + \frac{h_{k}^y}{2\beta_y}, \frac{1}{2} - \frac{h_{k+1}^y + s\sum_{i=1}^k h_{i}^y}{2\beta_y} \right)$. 
It can be easily checked that $0 \le q_{i}^x \le 1$ for all $(x,i)\in X\times[k]$, 
$\sum_{i=1}^{k+1} q_{i}^x = 1$, and $0 \le l_{i}^y \le 1$ for all $(y,i)\in Y\times[k+1]$. 
Moreover, 
\begin{eqnarray*} 
\sum_{i = 1}^{k+1} q_{i}^x l_{i}^y
&=& \sum_{i=1}^{k}\left(\frac{s+p_{i}^x}{\alpha_x}\right)\left(\frac{1}{2}+\frac{h_{i}^y}{2\beta_y}\right) + 
\frac{1}{\alpha_x}\left(\frac{1}{2} - \frac{h_{k+1}^y + s\sum_{i=1}^k h_{i}^y}{2\beta_y} \right)\\
&=& \frac{1}{2} + \frac{1}{2\alpha_x\beta_y}\left(\sum_{i=1}^k h_{i}^y p_{i}^x - h_{k+1}^y \right)
= 
\left\{
\begin{array}{l}
> 1/2 \mbox{~if~} \pmb{M}_f(x,y) = 1\\
< 1/2 \mbox{~if~} \pmb{M}_f(x,y) = -1.
\end{array}
\right.
\end{eqnarray*}
Hence, given a $k$-dimensional arrangement of points and hyperplanes 
 realizing  $\pmb{M}_f$, we can construct a protocol using at most $k+1$ 
 different messages for $f$. This means that $C(f) \le 
 \left\lceil\log{(k+1)}\right\rceil$. This completes the proof.
\end{proof}

Now we obtain our main result in this paper. 

\begin{theorem}\label{main-theorem}
For every function $f:X\times Y\rightarrow \{0,1\}$, $Q(f)=\lceil C(f)/2 \rceil$.
\end{theorem}

\section{Applications to Random Access Coding}

In this section we discuss the random access coding as an application of our characterizations 
of $Q(f)$ and $C(f)$. The concept of quantum random access coding (QRAC) and the classical random access coding 
(RAC) were introduced by Ambainis et al.\ \cite{ANTV99}. 
The $(m,n,p)$-QRAC (resp.\ $(m,n,p)$-RAC) is an encoding of $m$ bits 
using $n$ qubits (resp. $n$ bits) so that {\em any} one of the $m$ bits 
can be obtained with probability at least $p$. In fact, the function 
computed by the RAC (or QRAC) is known before as the {\it index} 
function in the context of communication complexity. It is denoted as 
$\mbox{{\it INDEX}}_n(x,i)=x_i$ for any 
$x\in\{0,1\}^n$ and $i\in [n]$ (see \cite{KN97}). 

\subsection{Existence of QRAC and RAC}
First we use Theorems \ref{qcc} and \ref{ccc} to show the existence of RAC and QRAC. 
As seen in \cite{PS86}, the smallest dimension of arrangements realizing 
$\mbox{\it {INDEX}}_n$ is $n$. Thus, Theorem \ref{qcc} gives us the 
following corollary for its unbounded-error one-way quantum communication 
complexity. 

\begin{corollary}
$Q(\mbox{\it {INDEX}}_n)=\lceil\log(n+1)/2\rceil$.
\end{corollary}

Similarly, Theorem \ref{ccc} gives its classical counterpart, which is tighter than \cite{PS86}. 

\begin{corollary}
$C(\mbox{{\it INDEX}}_n)=\lceil\log(n+1)\rceil$. 
\end{corollary}

Since random access coding is the same as $\mbox{{\it INDEX}}_n$ as 
Boolean functions, the following tight results are obtained for the existence of random access coding. 

\begin{corollary}\label{qrac-cor}
$(2^{2n}-1,n,>1/2)$-QRAC exists, but $(2^{2n},n,>1/2)$-QRAC does not 
 exist. Moreover, $(2^n-1,n,>1/2)$-RAC exists, but $(2^n,n,>1/2)$-RAC 
 does not exist.   
\end{corollary}

Corollary \ref{qrac-cor} solves the open problem in \cite{HINRY06} in 
its best possible form. It also implies the non-existence of 
$(2,1,>1/2)$-RAC shown in \cite{ANTV99}. Note that this fact does not 
come directly from the characterization of $C(f)$ in \cite{PS86}. 

\subsection{Explicit Constructions of QRAC and RAC}
In this subsection, we give an explicit construction 
of $(2^{2n}-1,n,>1/2)$-QRAC and $(2^n-1,n,>1/2)$-RAC that 
leads to a better success probability than what obtained from direct applications 
of Theorems \ref{qcc} and \ref{ccc}. For the case of QRAC, the construction 
is based on the proof idea of Theorem \ref{qcc} combined with the 
property of the index function. Their proofs are omitted due to space constraint.

\begin{theorem}\label{lowerbound-qrac}
For any $n \ge 1$, there exists a $(2^{2n}-1,n,p)$-QRAC 
such that $p \ge \frac{1}{2} + \frac{1}{2\sqrt{(2^n-1)(2^{2n}-1)}}$.
\end{theorem}

We can also obtain the upper bound of the success probability of 
$(2^{2n}-1,n,p)$-QRAC from the asymptotic bound by Ambainis et al. 
\cite{ANTV99}: 
For any $(2^{2n}-1,n,p>1/2)$-QRAC, $p\leq \frac{1}{2}+\sqrt{\frac{(\mathrm{ln}2)n}{2^{2n}-1}}$. 

It remains open to close the gap between the lower bound, $\approx 
1/2+\Omega(1/2^{1.5n})$, 
and the upper bound, $\approx 1/2+O(\sqrt{n}/2^n)$, of the success probability 

Similarly, for the case of RAC we have the following theorem. 

\begin{theorem}\label{prac}
There exists a $(2^n-1,n,p)$-RAC such that $p \ge \frac{1}{2} + \frac{1}{2 (2^{n+1}-5)}$.
\end{theorem}

The success probability of $(2^n-1,n,p)$-RAC can also be bounded
by the asymptotic bound in \cite{ANTV99} 
: For any $(2^{n}-1,n,p>1/2)$-QRAC, $p\leq \frac{1}{2}+\sqrt{\frac{(\mathrm{ln}2)n}{2^n-1}}$.

\section*{Acknowledgements}
R.R. would like to thank David Avis of McGill Univ. for introducing the 
world of arrangements, and Kazuyoshi Hidaka and Hiroyuki Okano of Tokyo 
Research Lab. of IBM Japan for their supports. We also thank Tsuyoshi 
Ito of Univ. of Tokyo and Hans Ulrich Simon of Ruhr-Universit\"{a}t 
Bochum for helpful discussion.


\begin{thebibliography}{Gur91}


\bibitem{Aar04}
S. Aaronson. Limitation of quantum advice and one-way communication. 
{\em Theory of Computing} {\bf 1} (2005) 1--28. 

\bibitem{Aar06}
S. Aaronson, The learnability of quantum states, quant-ph/0608142.

\bibitem{AA05}
S. Aaronson and A. Ambainis. Quantum search of spatial regions. {\em Theory of Computing} {\bf 1} (2005) 47--79. 

\bibitem{ANTV99}
A. Ambainis, A. Nayak, A. Ta-shma and U. Vazirani. 
Dense quantum coding and a lower bound for 1-way quantum automata. 
{\em Proc. 31st STOC}, pp.~376--383, 1999. Journal version appeared in {\em J. ACM} {\bf 49} (2002) 496--511. 

\bibitem{BJK04} Z. Bar-Yossef, T. S. Jayram, and I. Kerenidis, Exponential separation of quantum 
and classical one-way communication complexity. {\em Proc.\ 36th STOC}, pp.~128--137, 2004. 

\bibitem{BCW98} H. Buhrman, R. Cleve and A. Wigderson. Quantum vs. classical communication and computation. 
{\em Proc. 30th STOC}, pp.~63--68, 1998. 

\bibitem{BCWW01} H. Buhrman, R. Cleve, J. Watrous and R. de Wolf. 
Quantum fingerprinting. {\em Phys. Rev. Lett.} {\bf 87} (2001) Article no.\ 167902.

\bibitem{Fors02}
J. Forster. A linear lower bound on the unbounded error probabilistic communication complexity. 
{\em J. Comput. Syst. Sci.} {\bf 65} (2002) 612--625.

\bibitem{FKLMSS01}
J. Forster, M. Krause, S. V. Lokam, R. Mubarakzjanov, N. Schmitt, H. U. Simon. 
Relation between communication complexity, linear arrangements, and computational complexity. 
{\em Proc. 21st FSTTCS, Lecture Notes in Comput. Sci.} {\bf 2245} (2001) 171--182

\bibitem{FS06}
J. Forster and H. U. Simon. On the smallest possible dimension and the largest possible margin of linear arrangements 
representing given concept classes. {\em Theoret. Comput. Sci.} {\bf 350} (2006) 40--48. 

\bibitem{GKKRW06}
D. Gavinsky, J. Kempe, I. Kerenidis, R. Raz and R. de Wolf. 
Exponential separations for one-way quantum communication complexity, 
with applications to cryptography, to appear in {\em Proc. 39th STOC}. Also, quant-ph/0611209.

\bibitem{GKW06}
D. Gavinsky, J. Kempe and R. de Wolf. Strengths and weaknesses of quantum fingerprinting. 
{\em Proc. 21st CCC}, pp.~288--298, 2006. 

\bibitem{HINRY06}
M. Hayashi, K. Iwama, H. Nishimura, R. Raymond and S. Yamashita. 
$(4,1)$-quantum random access coding does not exist -- one qubit is not enough to recover one of four bits. 
{\em New J. Phys.} {\bf 8} (2006) Article no.\ 129.  

\bibitem{JS01}
L. Jak\'{o}bczyk and M. Siennicki. Geometry of Bloch vectors in two-qubit system. 
{\em Phys.~Lett.~A} {\bf 286} (2001) 383--390. 

\bibitem{KK04}
G. Kimura and A. Kossakowski. The Bloch-vector space for $N$-level systems 
-- the spherical-coordinate point of view. 
{\em Open Sys. Information Dyn.} {\bf 12} (2005) 207--229. Also, quant-ph/0408014.

\bibitem{Kla00}
H. Klauck. On quantum and probabilistic communication: 
Las Vegas and one-way protocols. {\em Proc. 32nd STOC}, pp.~644--651, 2000.

\bibitem{Kla01}
H. Klauck. Lower bounds for quantum communication complexity. 
{\em Proc. 42nd FOCS}, pp.~288-297, 2001.

\bibitem{KN97}
E. Kushilevitz and N. Nisan. {\em Communication Complexity}. Cambridge, 1997.

\bibitem{New91}
I. Newman. Private vs. common random bits in communication complexity. 
{\em Inform. Process. Lett.} {\bf 39} (1991) 67--71.

\bibitem{Nay99}
A. Nayak. Optimal lower bounds for quantum automata and random access codes. 
{\em Proc. 40th IEEE FOCS
}, pp.~369--376, 1999. 

\bibitem{NS06}
A. Nayak and J. Salzman. Limits on the ability of quantum states to convey classical messages. 
{\em J. ACM} {\bf 53} (2006) 184--206.

\bibitem{NC00} 
M. A. Nielsen and I. L. Chuang. {\em Quantum Computation and Quantum Information}, Cambridge, 2000. 

\bibitem{PS86}
R. Paturi and J. Simon. Probabilistic communication complexity. 
{\em J. Comput. Syst. Sci.} {\bf 33} (1986) 106--123.

\bibitem{Razb03}
A. Razborov. Quantum communication complexity of symmetric predicates. 
{\em Izvestiya Mathematics} {\bf 67} (2003) 145--159.

\end{thebibliography}
\end{document}